\documentclass{amsart}
\usepackage{epsfig}
\usepackage{latexsym}
\usepackage{amsmath}
\usepackage{graphicx}
\usepackage{amsmath}
\usepackage{amssymb}
\newtheorem{thm}{Theorem}

\theoremstyle{definition}

\newcommand{\ignore}[1]{}

\newcommand{\PP}{\mathbb{P}}

\begin{document}

\title{The optimal rate for resolving a near-polytomy in a phylogeny}
\author{Mike Steel$^1$ and Christoph Leuenberger$^2$}
\maketitle

\noindent {\em Affiliations:}
$^1$Biomathematics Research Centre, University of Canterbury, 8041, Christchurch, New Zealand;
$^2$D{\'e}partement de math{\'e}matiques, Universit{\'e} de Fribourg, Chemin du Mus{\'e}e 3, 1705 Fribourg, Switzerland.
\bigskip

\noindent{\em Email:} mike.steel@canterbury.ac.nz, christoph.leuenberger@unifr.ch

\noindent{\em Corresponding author:} Mike Steel 

\dedicatory{\today}

\begin{abstract}
The reconstruction of phylogenetic trees from discrete character data  typically relies on models that assume the characters evolve under a continuous-time Markov process operating at some overall rate $\lambda$.  When $\lambda$ is too high or too low, it becomes difficult to distinguish a short interior edge from  
a polytomy (the tree that results from collapsing the edge). 
In this note, we investigate the rate that maximizes the expected log-likelihood ratio (i.e. the Kullback--Leibler separation) between the  four-leaf unresolved (star) tree and a four-leaf binary tree with interior edge length $\epsilon$.  For a simple two-state model, we show that  as $\epsilon$ converges to $0$ the optimal rate also converges to zero when the four pendant edges have equal length. However, when the four pendant branches have unequal length, two local optima can arise, and it is possible for the globally optimal rate to converge to a non-zero constant as $\epsilon \rightarrow 0$.   Moreover, in the setting where the four pendant branches have equal lengths and either (i) we replace the two-state model by an infinite-state model or (ii) we retain the  two-state model and replace the
Kullback--Leibler separation by Euclidean distance as the maximization goal, then the optimal rate also converges to a non-zero constant.

\end{abstract}
\newpage

\section{Introduction}

When discrete characters evolve on a phylogenetic tree under a continuous-time Markov process, the states at the leaves provide information about the identity of the underlying tree.   It is known that when the overall substitution rates becomes  too high or too low, it becomes increasingly impossible to distinguish the tree from a less resolved tree (or indeed from any other tree) using  any given number of characters.

In particular, suppose we take a tree $T$ with an interior edge $e$ of length $\epsilon$ and we search for an overall substitution rate $\lambda_\epsilon$ that optimally discriminates (under some metric or criterion)  between $T_\epsilon$ and $T_0$ (i.e. the tree that has the same topology and branch lengths as $T_\epsilon$ except that $e$ has been collapsed (i.e. has length $0$)).  This optimal rate depends in an interesting way on the tree's branch lengths (and the metric or criterion used), as revealed by several studies over the last two decades (see, for example, \cite{fis09, lew16, tow07, tow11, yan98}), and applied to the study of data sets
(see, for instance, \cite{kop10, tow12}).  

In this short note, we consider a more delicate question that leads to some curious subtleties in its answer. Namely, how does $\lambda_\epsilon$ behave as
$\epsilon$ tends to zero? For simplicity, we consider  the four-leaf tree and two simple substitution models. We find that the answer to this question 
depends rather crucially on three things: whether the state space is finite or infinite, the metric employed, and the degree of imbalance in the branch lengths. 
Our results provide some analytic insight into simulation-based findings reported by  \cite{kop10} (in the second part of their section entitled `Optimum Rates of Evolution'); specifically, the optimal rate in the finite-state setting can behave differently from the optimal rate for generating characters that are parsimony-informative and homoplasy-free.

\section{Optimal rate results}

Consider a binary phylogenetic tree $T_\epsilon$ with four pendant edges of length $L$ and an interior edge of length $\epsilon$, as shown in Fig.~\ref{fig1}(i).
  Now consider a Markovian process that generates states at the leaves of $T_\epsilon$. We consider two models in this paper:  (a) the two-state symmetric model (sometimes referred to as the Neyman two-state model or the Cavender-Farris-Neyman model), and (b) an infinite-allele model (in which a change of state always leads to a new state, a model often referred to as the infinite alleles model of Crow and Kimura \cite{kim64}, or the random cluster model \cite{ste16}).  For both models the induced partition of the leaf set (in which the blocks are the subsets of leaves in the same state) will be referred to as a {\em character}.  Thus for the two-state model, there are exactly eight possible characters that can arise on $T_\epsilon$, while for the infinite-allele model, there are 13 when $\epsilon>0$ or 12 when $\epsilon = 0$ (there are 15 partitions of the set of four leaves of $T_\epsilon$, however when $\epsilon>0$ (resp. $\epsilon = 0)$ two (resp. three) have zero probability of being generated).

Suppose that  the branch lengths are all multiplied by a rate factor $\lambda \geq 0$, and let $P_\epsilon$ be the probability distribution on characters. 
Let  $P_0$ be the probability distribution on characters under the corresponding model on the star tree $T_0$ (shown in Fig.~\ref{fig1}(ii)).

\begin{figure}[htb]
  \begin{center}
  \includegraphics[width=7cm]{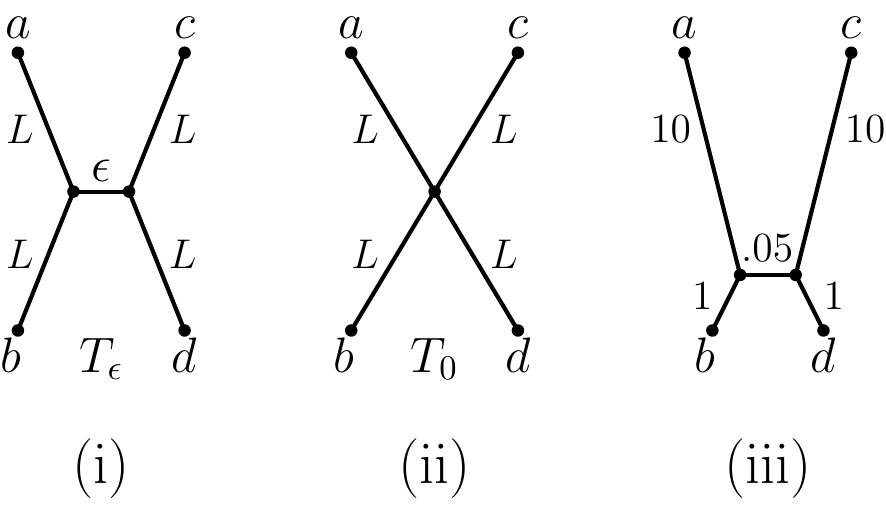}
  \end{center}
  \caption{(i) A binary four-leaf  tree $T_\epsilon$ with a short interior edge of length $\epsilon$ and four pendant edges of equal length $L$.  (ii) The star tree obtained from
  $T_\epsilon$ by setting $\epsilon=0$. (iii) A tree exhibiting two local optima for the rate $\lambda$ maximizing $d_{{\rm KL}}(P_0, P_\epsilon)$. }
  \label{fig1}
\end{figure}

Now, suppose that a data set $D$ of $k$ characters is generated by an independent and identically distributed (i.i.d.) process on the (unresolved) star  tree  $T_0$ (under either Model (a) or Model (b)). 
Let  ${\rm LLR}$ denote the log-likelihood ratio of the star tree  $T_0$ to the resolved tree $T_\epsilon$ (i.e. the logarithm of the 
ratio $\PP(D|T_0)/\PP(D|T_\epsilon)$).   As $k$ grows, $\frac{1}{k}{\rm LLR}$ converges in probability to its (constant) expected value,  which is precisely the  Kullback--Leibler separation (see \cite{cov06}):
 $$d_{{\rm KL}}(P_0, P_\epsilon)= \sum_{i} P_0(i) \ln \left(\frac{P_0(i)}{P_\epsilon(i)}\right),$$
 where the summation is over all the possible characters.

 Let $\lambda_\epsilon$ be a value of $\lambda$ that maximizes $d_{{\rm KL}}(P_0, P_\epsilon)$. From the previous paragraph, this is the rate that provides the largest expected  likelihood ratio in favour of the generating tree $T_0$ over an alternative resolved tree with an internal edge of length $\epsilon$.   We are interested in what happens to $\lambda_\epsilon$ as $\epsilon$ tends to zero. In that case, $d_{{\rm KL}}(P_0, P_\epsilon)$ also
converges to zero, but  it is not immediately clear whether the optimal rate that helps to distinguish $T_0$ from $T_\epsilon$ by maximizing $d_{{\rm KL}}(P_0, P_\epsilon)$ should be increasing, decreasing or converging to some constant value. A large rate improves the probability of a state-change occurring  on the central edge of $T_\epsilon$ however,
this comes at the price of greater randomization on the pendant edges, which tends to obscure the signal of such a change based on just the states at the leaves. 

It turns out that the limiting behaviour of
$\lambda_\epsilon$ depends crucially on whether the state space is finite or infinite. In Part (i) of the following theorem, we consider just the 
two-state symmetric model (see e.g. Chapter 7 of \cite{ste16}) 
but we indicate in Fig.~\ref{fig5} that a similar result appears to hold for the symmetric model on any number of states.   
The result in Part (i) contrasts with that in Part (ii) for the infinite-allele model, in which homoplasy (i.e. substitution to a state that has appeared elsewhere in the tree) does not arise. This second result is different from (but consistent with) a related result in \cite{tow07}.

\begin{thm}
\label{thm1}
\mbox{}
\begin{itemize}
\item[(i)]
For the two-state symmetric model, 
$$\lim_{\epsilon \rightarrow 0} \lambda_\epsilon = 0.$$
\item[(ii)]
By contrast, for the infinite-allele model,
$$\lim_{\epsilon \rightarrow 0} \lambda_\epsilon = \frac{1}{4L}.$$
\end{itemize}.
\end{thm}

\begin{proof}
{\em Part (i)}
Let  $p = \frac{1}{2}(1-\exp(-2 \lambda\epsilon))$ be the probability of a state change across the interior edge of $T_\epsilon$ under the two-state model.
Let $p_1$ be the probability of generating a character on $T_\epsilon$, where one leaf is in one partition block and the other three leaves are in a different partition block, and let $q_1$ be the corresponding probability on $T_0$. Because the four pendant edges of $T_\epsilon$ and $T_0$ have equal length, we have:
\begin{equation}
\label{eq1}
p_1 = (1-p) q_1 + p q_1 = q_1,
\end{equation}
so $q_1\ln(q_1/p_1)=0$.

Let $p_2$ be the  probability of generating either one of the two characters that have a  parsimony score of 2 on $T_\epsilon$, and let $q_2$ be the corresponding probability on $T_0$.  Once again we have:
\begin{equation}
\label{eq2}
p_2 = (1-p) q_2 + p q_2 = q_2,
\end{equation}
and so $q_2\ln(q_2/p_2)=0$.

 Let $p_{12}$ be the  probability of generating the character that has a  parsimony score of 1 on $T_\epsilon$ and a parsimony score of 2 on $T_0$, and let $q_{12}$ be the corresponding probability on $T_0$.   Let $p_0$ be the  probability of generating the character that has  parsimony score 0 on $T_\epsilon$  and let $q_0$ be the corresponding probability on $T_0$.
 Notice that we can write:
 \begin{equation}
 \label{alpha}
 q_0 = \alpha^4 + (1-\alpha)^4 \mbox{ and } q_{12} = 2\alpha^2(1- \alpha)^2,
 \end{equation}
 where $\alpha = \frac{1}{2}(1-e^{-2\lambda L})$.
Moreover, 
 $$p_{12} = (1-p)q_{12} + p q_0,$$
 and
 $$p_{0} = (1-p)q_0 + p q_{12}.$$
  It follows that 
\begin{equation}
\label{eq3}
q_{12}\ln(q_{12}/p_{12}) = -q_{12}\ln(p_{12}/q_{12}) =  -q_{12}\ln(1-p + p q_0/q_{12})
\end{equation}
 and 
 \begin{equation}
\label{eq4}
q_0\ln(q_0/p_0) =-q_0 \ln(p_0/q_0) = -q_0\ln(1-p + pq_{12}/q_0).
\end{equation}
 Combining Eqns. (\ref{eq1})--(\ref{eq4})  gives:
$$ d_{{\rm KL}}(P_0, P_\epsilon)= 0 + 0 - q_{12}\ln\left(1- p + \frac{q_0}{q_{12}}p\right) - q_0\ln\left(1- p + \frac{q_{12}}{q_0}p\right).$$
If we let $\theta = \theta(\lambda)  =q_{12}/q_0$,  then:
 \begin{equation}
 \label{KLfinal}
 d_{{\rm KL}}(P_0, P_\epsilon) =- q_0 \left(\theta \ln(1+(1-\theta)\frac{p}{\theta}) + \ln(1-(1-\theta) p) \right).
 \end{equation}
Notice that, by Eqn.~(\ref{alpha}),  we have:
 $$\theta = \frac{2\alpha^2(1-\alpha)^2}{\alpha^4+(1-\alpha)^4} = \frac{2(1-e^{-4\lambda L})^2}{(1+e^{-2\lambda L})^4 + (1-e^{-2\lambda L})^4},$$ and so, in Eqn.~(\ref{KLfinal}), $\theta$ is a monotone increasing function from $0$ (at $\lambda =0$) to a limiting value of $1$ as $\lambda \rightarrow \infty$.
Note also that $q_0 = q_0(\lambda)$ is a monotone decreasing function from $1$ (at $\lambda=0$) to a limiting value of  $\frac{1}{8}$ as $\lambda \rightarrow \infty$. 
 
Now let us set  $\lambda = x \epsilon$ in Eqn.~(\ref{KLfinal}), for a fixed value of $x$.  Then
$$p = \frac{1}{2}(1-\exp(-2\epsilon^2 x)) = x\epsilon^2 + O(\epsilon^3),$$ and from Eqn.~(\ref{alpha}) we have
$\theta = 2x^2L^2\epsilon^2 + O(\epsilon^3).$ 
Therefore:
\begin{equation}
\label{ello}
\ell(x):= \lim_{\epsilon \rightarrow 0} d_{{\rm KL}}(P_0, P_\epsilon)/\epsilon^2 =- 2x^2L^2 \ln\left(1+ \frac{1}{2xL^2}\right) + x,
\end{equation}
and so $\ell(x)$ converges to $\frac{1}{4L^2}$ as $x \rightarrow \infty$. 

Next, suppose that $\lambda_\epsilon$ does not converge to zero as $\epsilon \rightarrow 0$. Then for some $\delta>0$ and some sequence
of values $\epsilon_i$ which converges to zero, we have:
\begin{equation}
\label{lambdaeq}
\lambda_{\epsilon_i}> \delta >0
\end{equation}
 for  all $i$.
 Let $\lambda_i := \lambda_{\epsilon_i}, \theta_i:= \theta(\lambda_i)  \mbox{ and } p_i := \frac{1}{2} (1-e^{-2\lambda_i\epsilon_i}).$
Notice that $(1-\theta_i)p_i$  converges to zero as $i \rightarrow \infty$. This is because we can write
$0 \leq (1-\theta_i) p_i  \leq A e^{-B\lambda_i}(1-e^{-2\lambda_i\epsilon_i})$ for constants $A,B>0$, and differential calculus shows that the maximal value of
$A e^{-B\lambda}(1-e^{-2\lambda\epsilon})$ as $\lambda>0$ varies converges to zero as $\epsilon \rightarrow 0$. 
Since  $\theta_i$  is bounded away from 0 (by Inequality~(\ref{lambdaeq})), 
it also follows that $(1/\theta_i - 1)p_i = (1-\theta_i)p_i/\theta_i$  converges to zero as $i \rightarrow \infty$.

Consequently, 
 both  $(1-\theta_i)p$ and $(1-\theta_i)p/\theta_i$ will both  lie within $(0,1)$ for all $i \geq I$ for some sufficiently large finite value $I$ (dependent on $\delta$).  We now apply the following inequality and expansion which hold for all $x,y \in (0,1)$: 
$$ -\ln(1+x) < -x +\frac{x^2}{2},  \mbox{ and }  -\ln(1-y) = \sum_{j\geq 1} \frac{y^j}{j}$$
with $x=(1-\theta_i)p_i/\theta_i$ and $y=(1-\theta_i)p_i$ in Eqn.~(\ref{KLfinal}). Noting that the two linear terms in $p_i$ from Eqn.~(\ref{KLfinal}) cancel we obtain only quadratic and higher terms in $p_i$. Thus, for all $i \geq I$: 
\begin{equation}
\label{betterbound}
d_{{\rm KL}}(P_0, P_\epsilon) < q_0(\lambda_i)\left[ \frac{p_i^2}{2}(1-\theta_i)(\frac{1}{\theta_i} -\theta_i) + \sum_{j \geq 3} \frac{(1-\theta_i)^jp_i^j}{j}\right].
\end{equation}

Moreover, $q_0(\lambda_i) < q_0(0) = 1$ (by Eqn.~(\ref{lambdaeq})) and $p_i \leq  \lambda_i \epsilon_i$ for all $i$,   so we can write:
\begin{equation}
\label{bestbound}
\frac{d_{{\rm KL}}(P_0, P_\epsilon)}{\epsilon_i^2} < \left[ \frac{\lambda_i^2}{2}(1-\theta_i)(\frac{1}{\theta_i} -\theta_i)\right] + \left[ \frac{\epsilon_i}{2}\sum_{j \geq 3} \frac{(1-\theta_i)^j\lambda_i^j\epsilon_i^{j-3}}{j}\right].
\end{equation}

Let $y(t)= \frac{2(1-e^{-2t})^2}{(1+e^{-t})^4  + (1-e^{-t})^4}$. It can be verified that
$\frac{t^2}{2}(1-y(t))\left(\frac{1}{y(t)} -y(t)\right) < 1$ for all $t>0$.
Applying this with $t = 2\lambda_i L$, we obtain the following bound on the first term in Inequality~(\ref{bestbound}):
$$\left[ \frac{\lambda_i^2}{2}(1-\theta_i)(\frac{1}{\theta_i} -\theta_i)\right] < \frac{1}{4L^2}.$$
In addition, the second term on the right in Inequality~(\ref{bestbound}) converges to zero 
as $i$ grows, since the summation term is absolutely bounded (note that $\lambda_i(1-\theta_i) \rightarrow 0$ as $\lambda_i \rightarrow \infty$)  and since the numerator term out front, $\epsilon_i$, converges to $0$.

In summary, for sufficiently large values of $i$, we have
$d_{{\rm KL}}(P_0, P_\epsilon)/\epsilon_i^2 < \frac{1}{4L^2}$. Thus,  by selecting $x$ sufficiently large we can ensure that $\ell(x)$ (given by Eqn.~(\ref{ello}), and which is based on a 
$\lambda$ value that converges to zero as $\epsilon \rightarrow 0$) takes a larger
value for $d_{{\rm KL}}(P_0, P_\epsilon)$ than the value $\lambda_i$. 
This completes the proof of Part (i).

For Part (ii), let $y = 1-\exp(-\lambda L)$ and let $\zeta = 1-\exp(-\lambda \epsilon)$; these are the probabilities of a state change on a pendant and the interior edge, respectively, in the infinite-allele model.
The 12 partitions of $\{a,b,c,d\}$ that can be generated with strictly positive probability on $T_0$ fall into five disjoint classes (and the probabilities of generating the partitions within a given class are the same).  We label these classes $C_1, \ldots, C_5$ where
$C_1 = \{abcd\}, C_2 =\{a|bcd, b|acd, c|abd, d|abc\}$,
$C_3 = \{ac|b|d, ad|b|c, bc|a|d, bd|a|c\}$,  $C_4 = \{ab|c|d, cd|a|b\}$
and $C_5=\{a|b|c|d\}$.
For $i=1, \ldots, 5$, let  $P_0[i]$ (resp. $P_\epsilon[i])$ be the probability that the generated partition lies in Class $i$ on $T_0$. 
We have: $$P_0[1] = (1-y)^4,  P_0[2] = 4y(1-y)^3, P_0[3] = 4y^2(1-y)^2, P_0[4] = 2y^2(1-y)^2, P_0[5] = 4y^3(1-y)+y^4.$$
Let $P_\epsilon[i]$ be the corresponding probabilities of $T_\epsilon$.
We can then write:
\begin{equation}                
\label{pzeta}
P_\epsilon[i] = (1-\zeta)P_0[i] + \zeta D_i,
\end{equation}
where $D_i$ is dependent only on $y$. More precisely,
\begin{equation}
\label{Deq}
D_1=D_2=D_3=0, \mbox{ }  D_4= 2e^{-2\lambda L} (1-e^{-2\lambda L}), \mbox{ } D_5 = (1-e^{-2\lambda L})^2.
\end{equation}
The expression for $D_4$ arises because when there is a state change across the interior edge of $T_\epsilon$ then a partition  in  $C_4$
 occurs precisely when there is no state change between the two leaves on one side of the edge (with probability $e^{-2\lambda L}$) and there is a state change between the two leaves on the other side of the edge 
(with probability $1-e^{-2\lambda L}$); the coefficient of 2 out front recognises that there are two ways that this can occur.  For $D_5$, a state change across the interior edge of $T_\epsilon$ leads to the partition in $C_5$ precisely if the leaves on one side of the interior edge are in different states, and so too are the leaves on the other side of the interior edges, and these two independent events have probability $1-e^{-2\lambda L}$. 

Now,  each partition within any given class  $C_i$ has the same probability of being generated on $T_0$ (moreover, the same statement applies for $T_\epsilon$ in place of $T_0$).  This allows us to write
$d_{\rm KL}$ as a sum of five terms (rather than 12) in the following way:
 \begin{equation}
\label{pzeta2}
d_{{\rm KL}}(P_0, P_\epsilon) =  - \sum_{i=1}^5 P_0[i] \ln (P_\epsilon[i]/P_0[i]) = -\sum_{i=1}^5 P_0[i] \ln \left(1-\zeta\left(\frac{P_0[i] -D_i}{P_0[i]}\right)\right).
\end{equation}
Now,  $T_\epsilon$ has one additional partition type that it can generate but $T_0$ can not, namely the partition $\{ab|cd \}$.  This partition is generated by $T_\epsilon$ with probability $\zeta e^{-4\lambda L}$ and so $\sum_{i=1}^5 P_\epsilon[i] = 1- \zeta e^{-4\lambda L}$.
By Eqn. (\ref{pzeta}), we obtain the identity 
$$1-\zeta e^{-4\lambda L}= \sum_{i=1}^5 P_\epsilon[i] = (1-\zeta)\sum_{i=1}^5 P_0[i] + \zeta\sum_{i=1}^5 D_i = \sum_{i=1}^5 P_0[i] - \zeta\sum_{i=1}^5 (P_0[i]-D_i),$$ 
and since $\sum_{i=1}^5 P_0[i] =1$ we deduce that:
\begin{equation}
\label{p0eq}
\sum_{i=1}^5 (P_0[i] -D_i) = e^{-4\lambda L},
\end{equation}
which is equivalent to the identity  $D_1+\cdots +D_5 = 1- e^{-4\lambda L}$ from Eqn.~(\ref{Deq}).

Now,  $-\ln(1-x) \geq x$ for all values of $x < 1$ and combining this with Eqns.~(\ref{pzeta2}) and (\ref{p0eq}) gives
\begin{equation}
\label{prx}
d_{\rm KL}(P_0, P_\epsilon) \geq \sum_{i=1}^5 P_0[i] \cdot \zeta \frac{(P_0[i]-D_i)}{P_0[i]} = \zeta  \sum_{i=1}^5(P_0[i]-D_i) = \zeta e^{-4\lambda L}.
\end{equation}
By Eqn.~(\ref{pzeta2}) we can write
$d_{\rm KL}(P_0, P_\epsilon) =- \sum_{i=1}^5 a_i\ln(1-\zeta b_i),$
where $a_i$ and $b_i$ are functions of $\lambda$ defined for $i=1,2,3,4, 5$ by 
$a_i= P_0[i]$, $b_1=b_2=b_3=1$ and 
$$b_4 = \frac{a_4-D_4}{a_4} = \frac{-2e^{-\lambda L}}{(1-e^{-\lambda L})}, \mbox{ } 
b_5 = \frac{-4e^{-2\lambda L}}{(1-e^{-\lambda L})(1+3e^{-\lambda L})}.$$

For $0\leq \epsilon<L$ we have $|\zeta b_i|<1$ for each value of $i$.  
This is clear for $i=1,2, 3$; the cases $i=4$ and $i=5$ require a little care as $b_4, b_5 \rightarrow -\infty$ as $\lambda \rightarrow 0$.
However, $\zeta$ also depends on $\lambda$ and for $0\leq \epsilon<\frac{1}{3}L$ we have 
$$|\zeta b_4| \leq \frac{2e^{-\lambda L}(1-e^{-\lambda \epsilon})}{(1-e^{-\lambda L})} < \frac{2}{3} \mbox{ and }
|\zeta b_5| \leq \frac{1-e^{-\lambda \epsilon}}{1-e^{-\lambda L}} < 1.$$ 

Thus we expand Eqn.~(\ref{pzeta2}) via its Taylor series and write $d_{{\rm KL}}(P_0, P_\epsilon) = \sum_{i=1}^5  a_i \sum_{j \geq 1} \frac{b_i^j \zeta^j}{j}$, and since the term for $j=1$ is the term 
$\zeta e^{-4\lambda L}$ appearing in the lower bound for  $d_{{\rm KL}}(P_0, P_\epsilon)$ (see Eqn.~((\ref{prx})),  we have:
\begin{equation}
\label{onetwobound}
0\leq d_{{\rm KL}}(P_0, P_\epsilon) - \zeta e^{-4\lambda L} \leq \sum_{i=1}^5  a_i \sum_{j \geq 2} \frac{|b_i|^j \zeta^j}{j}.
\end{equation}
Notice that the term on the right of this last inequality can be written as
$\zeta^2 \cdot \sum_{i=1}^5  a_i   |b_i|^2  \cdot \sum_{j \geq 2} \frac{(|b_i|\zeta)^{j-2}}{j}.$
For each $i=1, \ldots, 5$ the term  $a_i|b_i|^2$ is bounded above by a constant times $e^{-\lambda L}$ (this is clear for $i=1,2, 3$ and the above formulae for 
$b_4$ and $b_5$ ensure it also holds for $i=4, 5 $ as there is a term $(1-e^{-\lambda L})^2$ in $a_i$ to cancel this term in the denominator of $b_i^2$). 
Consequently, from (\ref{onetwobound}), we can write 
\begin{equation}
\label{twobound}
0 \leq d_{{\rm KL}}(P_0, P_\epsilon) - \zeta e^{-4\lambda L} \leq \zeta^2 Ce^{-c \lambda},
 \end{equation}
where $C, c$ are absolute  and strictly positive constants (not dependent on $\lambda$ or $\epsilon$). 

By differential calculus,  $\lambda'_\epsilon = \frac{1}{\epsilon}\ln\left(1+ \frac{\epsilon}{4L}\right)$ maximizes $\zeta e^{-4\lambda L}$, and 
$\lim_{\epsilon \rightarrow 0} \lambda'_\epsilon = \frac{1}{4L}$. 
Now for the value of $\lambda_\epsilon$ that maximizes $d_{{\rm KL}}(P_0, P_\epsilon)$, 
Eqn.~(\ref{twobound}) shows that $\epsilon\lambda_\epsilon$ must tend to zero as $\epsilon \rightarrow 0$, since otherwise there is a sequence
of values $\lambda_{\epsilon_i}$ which tends to infinity, which leads to values for $d_{{\rm KL}}(P_0, P_\epsilon)$ that are smaller than those
obtained by setting  $\lambda = \lambda'_\epsilon$ (due to the exponential terms in Eqn. (\ref{twobound}))  which contradicts the optimality assumption on $\lambda_\epsilon$.  Thus $\epsilon\lambda_\epsilon \rightarrow 0$
as $\epsilon \rightarrow 0$ which implies that  
$\lim_{\epsilon \rightarrow 0} d_{{\rm KL}}(P_0, P_\epsilon)/\epsilon = \lambda e^{-4\lambda L}$, and this is maximized when $\lambda = \frac{1}{4L}$.
This completes the proof of Part (ii). 
\end{proof}

Despite the contrast exhibited by Theorem~\ref{thm1} between the infinite-allele and two-state setting we have a curious correspondence between the models for the Euclidean metric (i.e. the $L^2$ metric), as the following result shows.

\begin{thm}
\label{thm2}
For the Euclidean metric $d_2$, the substitution rate value $\lambda$ that maximizes $d_2(P_0, P_\epsilon)$ for the two-state symmetric model is given by
$\lambda~= \frac{1} {2 \epsilon} \ln \left(1+\frac{\epsilon}{2L}\right)$, which converges to $\frac{1}{4L}$ as $\epsilon \rightarrow 0$.
\end{thm}
\begin{proof}
Using the Hadamard representation for the two-state symmetric model  \cite{hen89}, an associated inner product identity (Eqn.~7.28 in \cite{ste16}) shows that:
\begin{equation}
\label{euclideq}
d_2(P_0, P_\epsilon)= \frac{1}{\sqrt{2}}e^{-4L\lambda}(1-e^{-2\lambda\epsilon}).
\end{equation}
This function of $\lambda$ has a unique local maximum at $\lambda = \frac{1} {2 \epsilon} \ln \left(1+\frac{\epsilon}{2L}\right)$. Now, $\lim_{\epsilon \rightarrow 0} \frac{1} {2 \epsilon} \ln \left(1+\frac{\epsilon}{2L}\right) = \frac{1}{4L}$, and at this value of $\lambda$ we have
$\lim_{\epsilon \rightarrow 0} d_2(P_0, P_\epsilon)/\epsilon= \frac{1}{2\sqrt{2}}e^{-1}$. \end{proof}

Theorems~\ref{thm1} and ~\ref{thm2} are  illustrated in Fig. \ref{fig2} and and Fig. ~\ref{fig2a}. 
Here, the edge lengths of the tree are $L=1$ (for each of the four pendant edges) and $\epsilon = 0.05$ and $\epsilon=0.1$. 
The values were calculated using Eqns.~(\ref{KLfinal}) and (\ref{pzeta2})  (using {\em Maple}), and are consistent with the expressions used in the derivation and statement of Theorems~\ref{thm1} and \ref{thm2}. 
 
\begin{figure}[htb]
  \begin{center}
  \includegraphics[width=12cm]{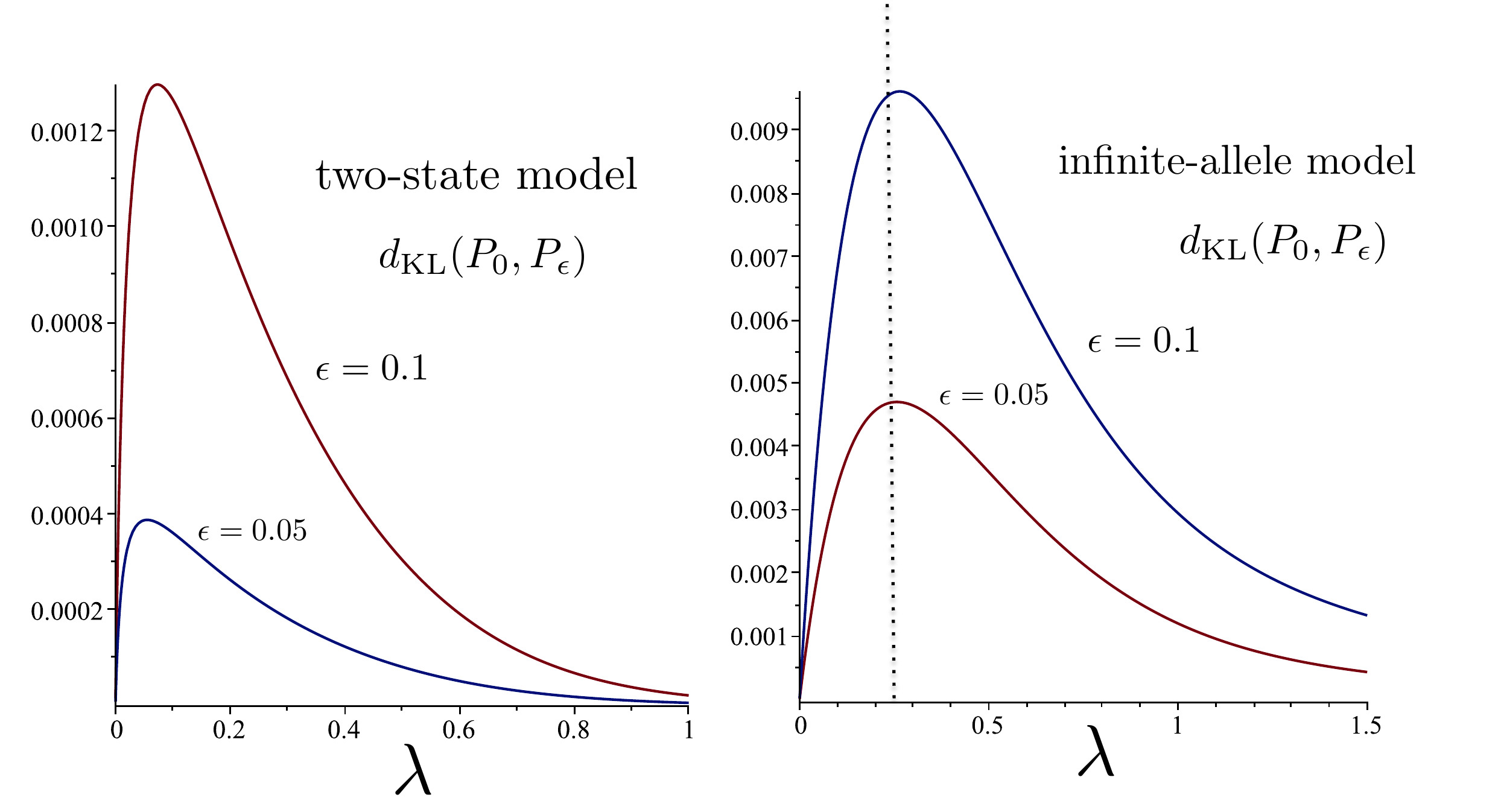}
  \end{center}
  \caption{Kullback--Leibler separation of $P_0$ and $P_\epsilon$ for a quartet tree with exterior edges of lengths $L=1$ and interior edge of length $\epsilon = 0.05$ and $\epsilon = 0.1$ as functions of $\lambda$, for the two-state and infinite-allele models (calculated using Eqns.~(\ref{KLfinal}) and (\ref{pzeta2}) respectively).  The graph on the left is consistent with a progression of the optimal $\lambda$ value towards zero as $\epsilon$ decreases. For the infinite-allele model (right),  the optimal $\lambda$ value converges to $\frac{1}{4L} = 0.25$ as $\epsilon \rightarrow 0$, and  $d_{{\rm KL}}(P_0, P_\epsilon)$ converges to $\frac{\epsilon}{4L}e^{-1}$.}  
  \label{fig2}
\end{figure}

\begin{figure}[htb]
  \begin{center}
  \includegraphics[width=7cm]{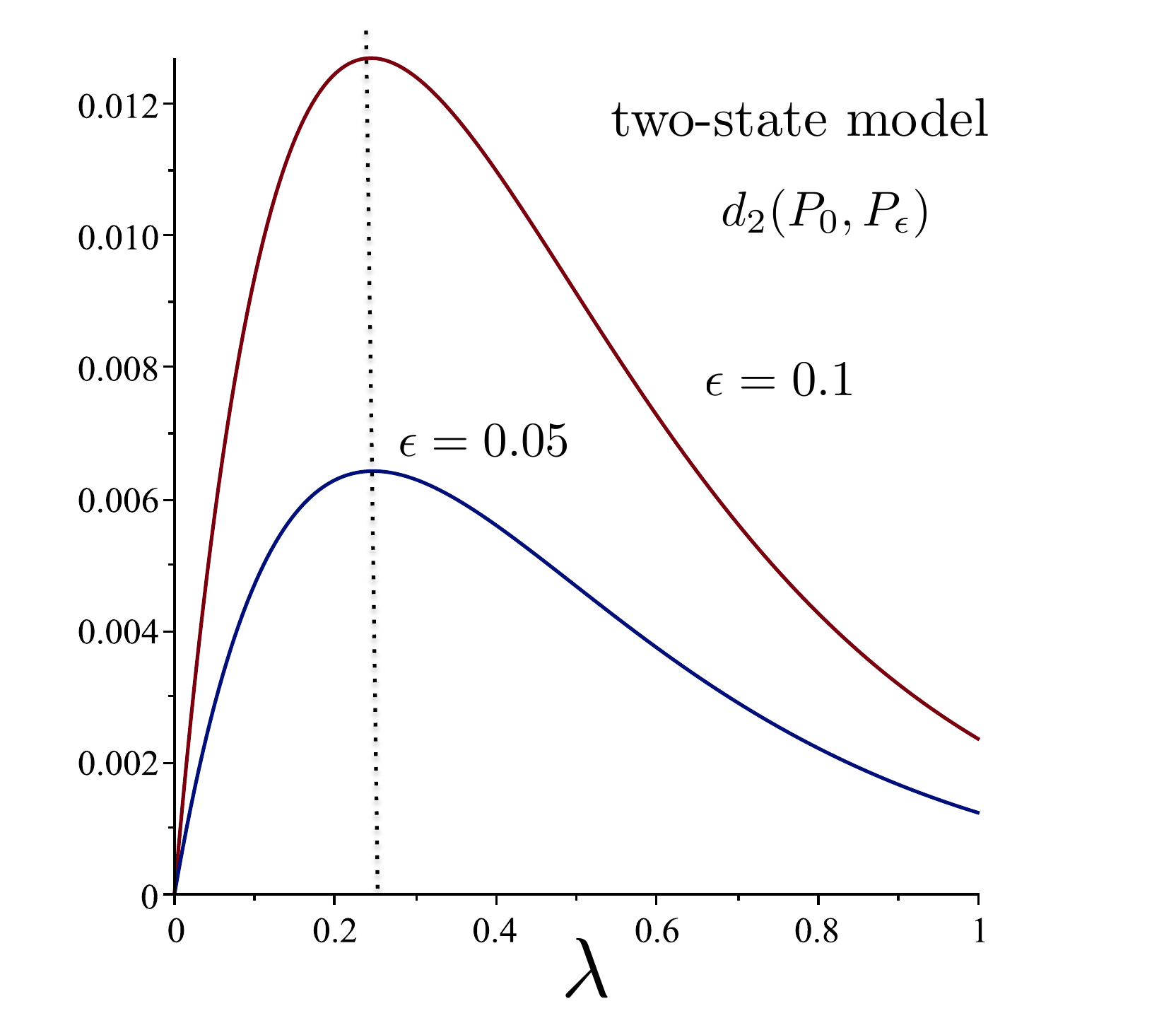}
  \end{center}
  \caption{The Euclidean distance $d_2$ between  $P_0$ and $P_\epsilon$ for a quartet tree with exterior edges of lengths $L=1$ and interior edge of length $\epsilon = 0.05$ and $\epsilon = 0.1$ as functions of $\lambda$, for the two-state model (calculated using Eqn.~\ref{euclideq}). The optimal $\lambda$ value converges to $\frac{1}{4L} = 0.25$ as $\epsilon \rightarrow 0$, giving $\frac{\epsilon}{2\sqrt{2}}e^{-1}$ as the asymptotic  value of $d_2(P_0, P_\epsilon)$.}     
  \label{fig2a}
\end{figure}

The Kullback--Leibler separation is closely related to the Fisher information  \cite{leh03}, and the usefulness of the Fisher information in phylogenetic trees has first been studied in \cite{gold}.  
It follows from large sample theory that the variance of an efficient estimator of the edge length $\epsilon$, based on the data set $D$ for a large number $k$ of characters (and with $L$ and $\lambda$ being known), is inversely proportional to the Fisher information with respect to the parameter $\epsilon$. The Fisher information is defined by:
\begin{equation*}
I(\epsilon) = - \mathbb{E}\left[\frac{d^2}{d\epsilon^2} \ln P_\epsilon(i)\right] =- \sum_i P_\epsilon(i) \frac{d^2}{d\epsilon^2} \ln P_\epsilon(i).
\end{equation*}
For the two-state symmetric model, we can
expand $\ln P_\epsilon(i)$ in a Taylor series around $\epsilon = 0$ to obtain:
\begin{equation}\begin{split}\label{FisherInfo}
d_{{\rm KL}}(P_0,P_\epsilon) & = \sum_i P_0(i) \ln\frac{P_0(i)}{P_\epsilon(i)} = \sum_i P_0(i) \left(\ln {P_0(i)} - \ln P_\epsilon(i)\right)\\
& = \sum_i P_0(i) \left(-\epsilon\frac{d}{d\epsilon}\bigg|_{\epsilon=0}\ln P_\epsilon(i) - \frac{1}{2}\epsilon^2\frac{d^2}{d\epsilon^2}\bigg|_{\epsilon=0}\ln P_\epsilon(i)\right) + O(\epsilon^3)\\
& = -\epsilon \frac{d}{d\epsilon}\bigg|_{\epsilon=0}\left(\sum_i P_\epsilon(i)\right) - \frac{1}{2}\epsilon^2 \sum_i P_0(i) \frac{d^2}{d\epsilon^2}\bigg|_{\epsilon=0}\ln P_\epsilon(i)+O(\epsilon^3)\\
& = \frac{1}{2}\epsilon^2I(0)+O(\epsilon^3).
\end{split}\end{equation}
In the two-state model with $L=1$, analysis of the coefficient of $\epsilon^2$ in Eqn.~(\ref{KLfinal}) (using {\em Mathematica}) provides the following explicit
description of the Fisher information term  $I(0)$:
\begin{equation}
I(0) = \frac{8\lambda^2\exp(-4\lambda)\cosh^2(2\lambda)}{(3+\cosh(4\lambda))\sinh^2(2\lambda)}.
\end{equation}

Let $\hat{\epsilon}$ be an asymptotically efficient estimator of the short edge length, i.e. one whose variance asymptotically achieves the Cram\'er-Rao bound (see Ch. 6.2 in \cite{leh03}). Then its relative error, based on a large number $k$ of i.i.d. generated characters, is roughly 
\begin{equation*}
\mbox{relative error}(\hat{\epsilon}) := \frac{\sqrt{\mbox{var}(\hat{\epsilon})}}{\epsilon} \approx \frac{\sqrt{1/k I(\epsilon)}}{\epsilon}.
\end{equation*}
From \eqref{FisherInfo}, we get, for small values of $\epsilon$, the following approximation under the two-state symmetric model:
\begin{equation*}
\mbox{relative error}(\hat{\epsilon}) \approx \left[2k \ d_{{\rm KL}}(P_0,P_\epsilon)\right]^{-1}.
\end{equation*}
From this we see that the optimal rate $\lambda_\epsilon$ minimizes the relative estimation error for the edge length, again underlying the usefulness of Kullback--Leibler separation in this setting.

\section{Unequal pendant edge lengths}

Theorem~\ref{thm1}(i) is not generally valid for quartets when we drop the assumption of equal edge lengths. Fig. \ref{fig3} shows the Kullback--Leibler separation $d_{{\rm KL}}(P_0,P_\epsilon)$ dependent on  $\lambda$ for a quartet tree with interior edge length $\epsilon = 0.05$ and unequal lengths of $1$ and $10$ on the pendant edges on one side of the interior edge $e$, and also $1$ and $10$ on the pendant edges on other side of $e$ (as shown in Fig.~\ref{fig1}(iii)). 

 There is still a local maximum which tends to $0$ as $\epsilon \to 0$ but the global maximum $\lambda_\epsilon$ stays bounded away from $0$.

\begin{figure}[htb]
  \begin{center}
  \includegraphics[width=7cm]{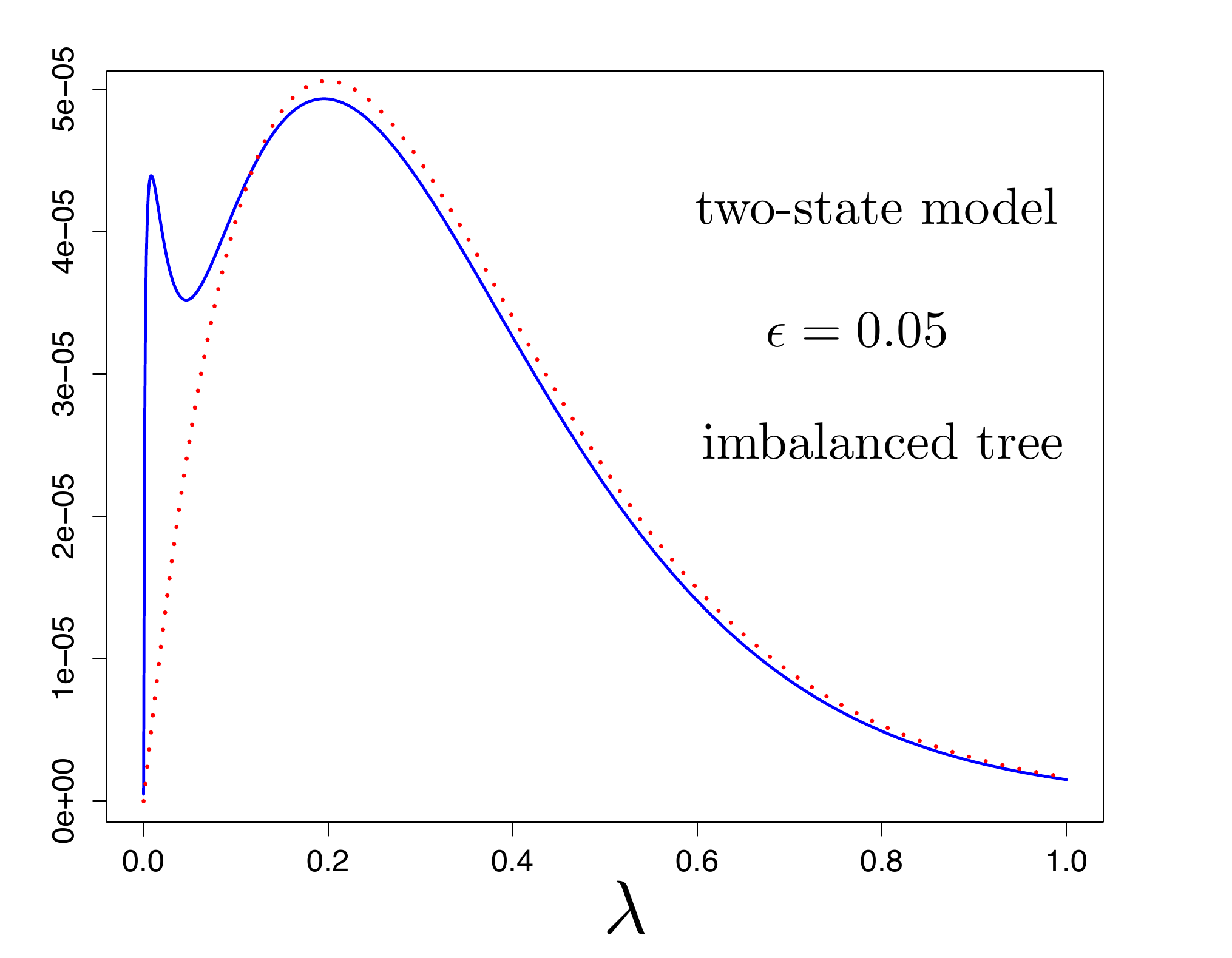}
  \end{center}
  \caption{Kullback--Leibler separation (solid line) for a quartet tree with exterior edges of lengths $1$ and $10$ on both sides of the interior edge of length $\epsilon = 0.05$ as a function of $\lambda$. The dotted line shows the Fisher information term from Eqn. \eqref{I_2}.}
  \label{fig3}
\end{figure}

This idiosyncratic shape of the curve in the highly asymmetric case can be explained as follows: If two edges on each side of the central edge are very long, they can essentially be ignored (the state at each of the two leaves is almost completely random) and the states at the leaves of the two shorter edges of lengths $1$ are more informative for inferring the total length of $2+\epsilon$ of the path joining them, and thereby for deciding whether or not  $\epsilon = 0$. Let
\begin{equation*}
p(\epsilon) = \frac{1}{2}\left(1+\exp(-2\lambda(2+\epsilon))\right)
\end{equation*}
be the probability that the two characters at the leaves of the unit-length edges are in the same state. Then the Fisher information with respect to $\epsilon$, when we ignore the characters at the leaves of the two long edges, is the given by:
\begin{equation}\begin{split}\label{I_2}
I_2(\epsilon) &= - p(\epsilon) \frac{d^2}{d\epsilon^2}\ln p(\epsilon) - (1-p(\epsilon)) \frac{d^2}{d\epsilon^2}\ln (1-p(\epsilon))\\
& = \frac{4\lambda^2}{\exp(4\lambda(2+\epsilon))-1}.
\end{split}\end{equation}
Fig. \ref{fig3} shows $\epsilon^2 I_2(0)/2$ as a function of $\lambda$ (the dotted line). Clearly, the global maximum is explained by the estimation of $\epsilon$ via the two unit-length edges.  As for Fig.~\ref{fig2}, the values were calculated by simulating characters  over a range of $\lambda$ values (using the R statistical package).

By  lessening the imbalance between the pendant edge lengths it is possible to make the two local optima for the rates have equal (global) optimal values; which provides an example where the global optimal rate for maximizing $d_{{\rm KL}}(P_0,  P_\epsilon)$ is not unique.  When the edge length imbalance decreases further, simulations suggest that Theorem~\ref{thm1}(i) remains valid (i.e.  the limit $\lambda_\epsilon \rightarrow 0$ as $\epsilon \rightarrow 0$ does not just hold for the special setting in which all pendant edges have exactly equal lengths).

\subsection{Concluding comments}

For the biologist, Theorem 1 (i) provides a caution:  in resolving a near-polytomy, it is tempting to search for genetic data that have evolved fast enough to have undergone substitution events on the interior edge; however, a slower-evolving data set may, in fact, be more likely to distinguish the resolved tree from an unresolved phylogeny.  For infinite-allele models, however, Theorem 1(ii) ensures there is  a positive optimal rate regardless of how short the interior edge is (consistent with a related result from \cite{tow07}).   Theorem 1 applies to balanced trees, and we also saw that for sufficiently unbalanced trees, these findings can change due to the appearance of a second local optimal rate that eventually becomes the global optimal rate ({\em cf} Fig.~\ref{fig3}). 
Moreover, as Fig~\ref{fig5} indicates, the results established for the two-state symmetric model appear to hold for other finite-state models such as the four-state symmetric model (often referred to as the `Jukes-Cantor (1969)' model (JC69); for details see \cite{ste16}, Section 7.2.2).

\begin{figure}[htb]
  \begin{center}
  \includegraphics[width=12cm]{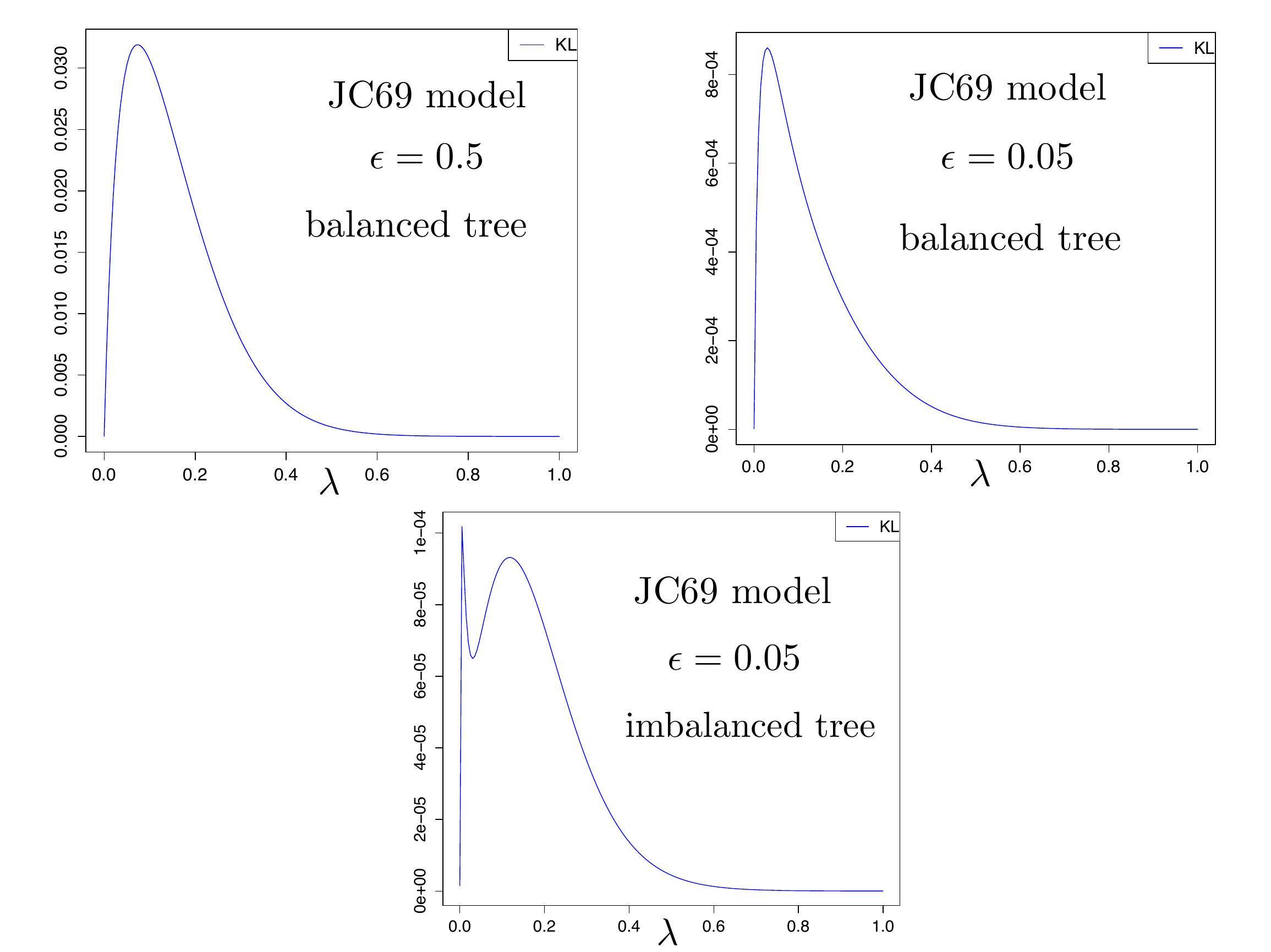}
  \end{center}
  \caption{The four-state symmetric model (JC69) shows similar behaviour to the two-state symmetric model. Top: Kullback--Leibler separation of $P_0$ and $P_\epsilon$ for a quartet tree with exterior edges of lengths $L=1$ and interior edge of length $\epsilon = 0.5$ and $\epsilon = 0.05$ as functions of $\lambda$. Bottom: Kullback--Leibler separation of $P_0$ and $P_\epsilon$ for a quartet tree with exterior edges of lengths $L=1$ and $L=10$ as shown in Fig.~\ref{fig1}(iii), and $\epsilon=0.05$.
}  
  \label{fig5}
\end{figure}

We have assumed throughout that the lengths of the four pendant edges of the tree remain fixed as the interior edge shrinks to zero. If the pendent edges are also allowed to shrink, then
Theorem 1 no longer applies.  For example, consider the binary tree $T'_\epsilon$ that has an interior edge of length $\epsilon$ and four  pendant edges of length $L\epsilon$. Then  as $\epsilon \rightarrow 0$,
the optimal rate $\lambda_\epsilon$ now {\em increases} towards infinity rather than decreasing to zero as $\epsilon \rightarrow 0$. The reason for this is quite simple. Consider the tree $T_1$ that has an interior edges of length 1 and four pendant edges each of length $L$. This tree has some optimal rate $\lambda^*$ that
maximizes $d_{{\rm KL}}$.  Now $T'_\epsilon$ is obtained from $T_1$ by multiplying each edge length of $T_1$ by $\lambda \epsilon$. 
Thus,  for $T_\epsilon$, the optimal rate is given by $\lambda_\epsilon =\lambda^*/ \epsilon \rightarrow \infty$ as $\epsilon \rightarrow 0$.

Finally,  we have considered $d_{{\rm KL}}(P_0, P_\epsilon)$ rather than $d_{{\rm KL}}(P_\epsilon, P_0)$,  partly because the former is easier to analyse mathematically, and is well-defined in the infinite-allele setting ($d_{{\rm KL}}(P_\epsilon, P_0)$ is not well-defined for this model since the partition $\{ab|cd \}$ has positive probability under $P_\epsilon$ for $\epsilon>0$ and zero probability under $P_0$). However, a more fundamental reason for our choice of $d_{{\rm KL}}(P_0, P_\epsilon)$  is that it is more natural to consider the unresolved tree ($\epsilon = 0$)  as the null hypothesis and the resolved tree as the alternative hypothesis because one typically wishes to disprove the null which in our case  means to reject the polytomy.  

\section{Acknowledgements}
We thank Jeffrey Townsend and Daniel Wegmann for helpful discussions. We also thank the three anonymous reviewers for numerous helpful suggestions that have improved the paper.


\begin{thebibliography}{10}
\bibitem{cov06} Cover, T. M. and Thomas, J. A. (2006).  Elements of Information Theory (2nd ed.) Wiley-Interscience. 
\bibitem{fis09}  Fischer, M. and Steel, M. (2009). Sequence length bounds for resolving a deep phylogenetic divergence. J. Theor. Biol. 256: 247--252.
\bibitem{gold} Goldman, N. (1998).  Phylogenetic information and experimental design in molecular systematics, Proc. Roy. Soc. B 265: 1779-1786.
\bibitem{hen89} Hendy, M. D. (1989). The relationship between simple evolutionary tree models and observable sequence data, Syst. Biol. 38: 310--321.
\bibitem{kim64} Kimura, M. and Crow, J (1964). The number of alleles that can be maintained in a finite population. Genetics. 49: 725--738.
\bibitem{kop10} Klopfstein, S., Kropf, C. and Quicke, D.L.J. (2010). An evaluation of phylogenetic informativeness profiles and the molecular phylogeny of Diplazontinae (Hymenoptera, Ichneumonidae), Syst. Biol. 59(2): 226--241. 
\bibitem{leh03} Lehmann, E. L. and Casella, G. (1998). Theory of point estimation (2nd ed.)  Springer.
\bibitem{lew16}  Lewis, P. O., Chen, M.-H., Luo, L., Lewis, L. A., Fu$\check{c}$ikov{\'a}, K., Neupane, S., Wang, Y.-B. and Shi, D. (2016). Estimating Bayesian phylogenetic information content. Syst. Biol. 65(6): 1009--1023.
\bibitem{ste16} Steel, M. (2016).  Phylogeny: discrete and random processes in evolution. CMBS-NSF Regional Conference Series in Applied Mathematics No. 89. SIAM Philadelphia PA.
\bibitem{tow07} Townsend, J.P. (2007). Profiling phylogenetic informativeness, Syst. Biol. 56(2): 222--231.
\bibitem{tow11} Townsend, J.P. and Leuenberger, C. (2011). Taxon sampling and the optimal rates of evolution for phylogenetic inference. 
Syst Biol. 60(3): 358--365.
\bibitem{tow12} Townsend J.P., Su Z, Tekle Y.L. (2012), Phylogenetic signal and noise: predicting the power of a data set to resolve phylogeny. Syst. Biol.
61(5): 835--849. 
\bibitem{yan98} Yang, Z. (1998),  On the best evolutionary rate for phylogenetic analysis. Syst. Biol.  47:125--133. 
 
\end{thebibliography}
\end{document}